\documentclass{article}

\PassOptionsToPackage{numbers,compress}{natbib}
\usepackage[final]{neurips_2025}

\usepackage{tabularx} % 在导言区添加
\usepackage[most]{tcolorbox}
\usepackage{amsmath}
\usepackage{subcaption} 
\usepackage[utf8]{inputenc} % allow utf-8 input
\usepackage[T1]{fontenc}    % use 8-bit T1 fonts
\usepackage{hyperref}       % hyperlinks
\usepackage{wrapfig}  % 用于文字环绕图片
\usepackage{url}            % simple URL typesetting
\usepackage{array} % 引入 array 宏包
\usepackage{tabularx} 

\usepackage{float}
\usepackage{booktabs}       % professional-quality tables
\usepackage{amsfonts}       % blackboard math symbols
\usepackage{nicefrac}       % compact symbols for 1/2, etc.
\usepackage{microtype}      % microtypography
\usepackage{xcolor}         % colors
\usepackage{enumerate}
\usepackage{enumitem}
\usepackage{hyperref}
\usepackage{amsmath}
\usepackage{amssymb}
\usepackage{amsthm}
\usepackage{authblk}  

\theoremstyle{plain}
\newtheorem{theorem}{Theorem}
\newtheorem{lemma}{Lemma}
\newtheorem{proposition}{Proposition}

\theoremstyle{definition}

\theoremstyle{remark}

\hypersetup{
    colorlinks=true,    % 启用彩色链接
    linkcolor=blue,     % 内部链接颜色（包括\ref）
    citecolor=red,    % 引用颜色
    urlcolor=magenta    % URL链接颜色
}

\title{ElasticMM: Efficient Multimodal LLMs Serving with Elastic Multimodal Parallelism}

\author[1,2]{\textbf{Zedong Liu}}
\author[3]{\textbf{Shenggan Cheng}}
\author[1]{\textbf{Guangming Tan}}
\author[3]{\textbf{\dag Yang You}} % 
\author[1]{\textbf{\dag Dingwen Tao}} % 

\affil[1]{Institute of Computing Technology, Chinese Academy of Sciences}
\affil[2]{University of Electronic Science and Technology of China}
\affil[3]{National University of Singapore}

\begin{document}
\footnotetext{$\dag$ Corresponding Authors.}

\maketitle

\begin{abstract}
Multimodal large language models (MLLMs) extend LLMs to handle images, videos, and audio by incorporating feature extractors and projection modules. However, these additional components—combined with complex inference pipelines and heterogeneous workloads—introduce significant inference overhead. Therefore, efficiently serving MLLMs remains a major challenge. Current tightly coupled serving architectures struggle to distinguish between mixed request types or adapt parallelism strategies to different inference stages, leading to increased time-to-first-token (TTFT) and poor resource utilization. To address this, we introduce Elastic Multimodal Parallelism (EMP), a new serving paradigm that elastically adapts to resource heterogeneity across request types and inference stages. Building upon EMP, we develop ElasticMM, an MLLM serving system that (1) separates requests into independent modality groups with dynamic resource allocation via a modality-aware load balancer; (2) decouples inference stages and enables parallelism adjustment and adaptive scaling via elastic partition scheduling; and (3) improves inference efficiency through unified multimodal prefix caching and non-blocking encoding. Experiments on diverse real-world datasets show that ElasticMM outperforms state-of-the-art (SOTA) serving systems, reducing TTFT by up to 4.2$\times$ and achieving 3.2–4.5$\times$ higher throughput while meeting service-level objectives (SLOs).
\end{abstract}

\section{Introduction}
\label{sec1}
Large Language Models (LLMs) have had a profound impact on real-world applications due to their exceptional performance~\citep{chen2022visualgpt,wu2023multimodal,liang2024survey}. Their gradual expansion into the multimodal domain has led to the emergence of Multimodal Large Language Models (MLLMs), which can process inputs such as images, videos, and audio~\citep{hu2023promptcap,mokady2021clipcap}. By employing feature extractors and projection modules to map multimodal inputs into the LLM feature space, MLLMs excel at tasks like visual question answering (VQA), image captioning, and multimodal interaction~\citep{schwenk2022okvqa,chen2024internvl,li2024llava,team2024gemini}. Consequently, MLLMs are increasingly deployed in online services, where they are highly sensitive to service-level objectives (SLOs).

Unlike traditional LLMs that process only text~\citep{vaswani2017attention,touvron2023llama,floridi2020gpt,achiam2023gpt,bai2023qwen}, MLLMs must handle diverse input types, introducing additional model components and a more complex inference pipeline. This pipeline typically includes distinct stages such as image preprocessing, image encoding, and language model inference (comprising prefill and decode). These added stages increase computational complexity (Fig.~\ref{fig1}b) and significantly extend the average context length (Fig.~\ref{fig1}c). However, existing inference serving systems adopt a tightly coupled architecture that neither distinguishes between request types nor decouples inference components, executing all stages on the same hardware instances~\citep{wolf2020transformers,kwon2023efficient,aminabadi2022deepspeed,zheng2024sglang}. As a result, time-to-first-token (TTFT) \textit{increases sharply} under heavy multimodal workloads. Furthermore, encoder-decoder models introduce cross-attention layers~\citep{alayrac2022flamingo,dai2024nvlm,chi2024llama}, and the computational heterogeneity between text-only and multimodal requests \textit{reduce batching efficiency}.

\begin{figure}
    \centering
    \vspace{-2mm}
    \begin{subfigure}{\linewidth}
        \includegraphics[width=\linewidth]{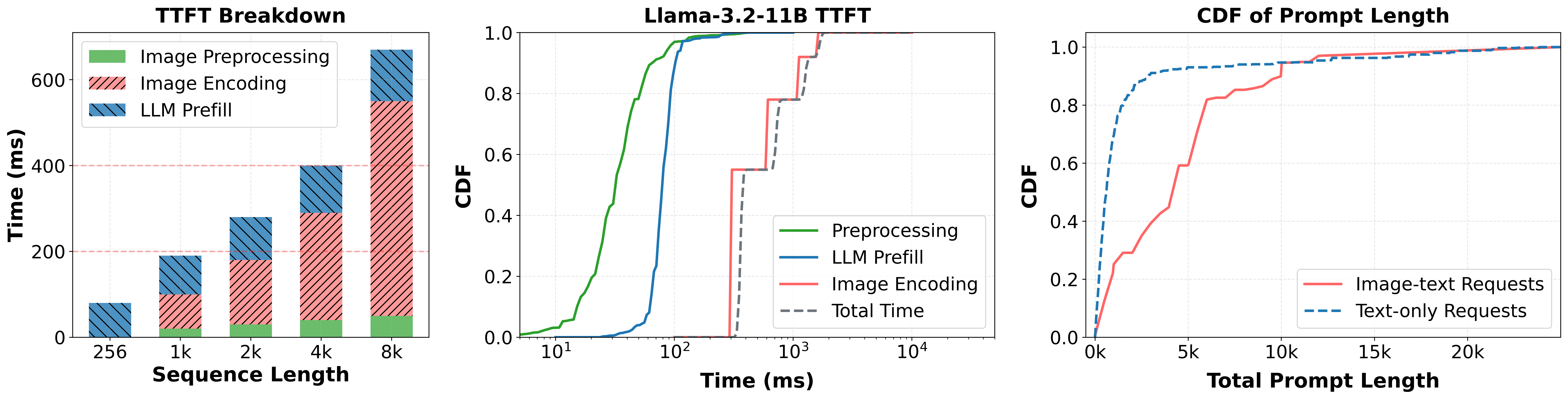}
        \caption*{} % 不自动编号，不显示文字
        \label{fig:combined}
    \end{subfigure}
    \vspace{-2mm}
    {\small
    \vspace{-5mm}
    \begin{tabular}{
        >{\centering\arraybackslash}m{0.25\linewidth}
        >{\centering\arraybackslash}m{0.38\linewidth}
        >{\centering\arraybackslash}m{0.3\linewidth}
    }
        \textbf{(a)}~TTFT breakdown\phantomsection\label{fig1a} &
        \textbf{(b)}~CDF of inference stage latencies\phantomsection\label{fig1b} &
        \textbf{(c)}~CDF of prompt lengths\phantomsection\label{fig1c} \\
    \end{tabular}
    }
    \vspace{-4mm}
    \caption{MLLMs' inference overhead and workload complexity. (a) and (b) demonstrate the significant overhead introduced by MLLMs. (c) reveals the longer context in multimodal requests. Results obtained using the LLaMA3.2-11B model on the ShareGPT-4o dataset.}
    \label{fig1}
    \vspace{-5mm}
\end{figure}

To address this issue, a decoupled multimodal inference architecture emerges as a promising solution—processing text-only and multimodal requests on separate instances. Our systematic analysis reveals that multimodal workloads exhibit bursty patterns, often characterized by sudden spikes in image inputs (as also observed in previous work~\citep{qiumodserve}), and that different inference stages benefit from varying degrees of parallelism. For example, image encoding and LLM prefill stages have higher computational complexity and benefit from large-scale parallelism, whereas the decode stage has limited scalability~\citep{zhong2024distserve,patel2024splitwise,agrawal2023sarathi}. These observations motivate two critical design challenges: (1) static resource allocation struggles to adapt to dynamically changing workloads, and (2) fixed parallelism strategies often fail to match the varying resource demands of different inference stages.

To this end, we propose \textbf{Elastic Multimodal Parallelism (EMP)}. To address the first challenge, we introduce a \textit{modality-aware load balancing technique} that monitors real-time workload fluctuations and dynamically adjusts resource allocation across different modality groups. To tackle the second challenge, we design an \textit{elastic partitioned scheduling technique}, which enables dynamic parallelism adaptation at the granularity of each inference stage. For example, compute-intensive stages such as encoding and prefill are scaled across more GPUs, while the decode stage reduces parallelism to free up resources for other requests.

Building on EMP, we present \textbf{ElasticMM}, an MLLM serving system with four key contributions:
\begin{itemize}[leftmargin=1.3em,topsep=0pt]
\item We identify tight coupling in existing systems as a major bottleneck for MLLM serving and propose EMP, a new serving paradigm that separates text-only and multimodal requests into independent modality groups and decouples inference stages for stage-specific parallelism.
\item We design two core techniques for EMP: (1) \textbf{\textit{modality-aware load balancing}}, which dynamically allocates and scales resources across modalities to handle unpredictable workloads; and (2) \textbf{\textit{elastic partitioned scheduling}}, which controls parallelism through flexible scheduling and scaling to maximize inference performance.
\item We propose two optimizations to mitigate MLLM-specific overheads: \textbf{\textit{unified multimodal prefix caching}} to reduce redundant computation and data transfer, and \textbf{\textit{non-blocking encoding}} to minimize the impact of encoding latency on the overall inference pipeline.
\item We conduct a comprehensive evaluation of ElasticMM on two real-world datasets. Compared to the SOTA baseline vLLM~\citep{kwon2023efficient}, ElasticMM reduces TTFT by up to 4.2$\times$ and achieves a 3.2$\times$–4.5$\times$ throughput improvement while meeting service-level objective (SLO) requirements.
\end{itemize}

\section{Background and Motivation}

\subsection{Multimodal Large Language Models Inference}

\textbf{MLLM Inference Pipeline.} MLLMs extend traditional language models by integrating visual, auditory, and other modalities, enabling unified reasoning across heterogeneous input sources for complex multimodal tasks~\citep{chen2022visualgpt,wu2023multimodal,liang2024survey,mokady2021clipcap}. Taking vision-language models as an example, the MLLM inference pipeline typically consists of three key stages:
(1) Image Preprocessing: Raw images are resized and divided into uniformly sized tiles~\citep{radford2021learning}.
(2) Image Encoding: A vision encoder extracts visual features and converts them into vision tokens~\citep{dosovitskiy2020image,zhai2023sigmoid}.
(3) Text Generation: A language model takes both vision tokens and a text prompt as input to generate responses.

\textbf{MLLM Architectures.} Modern MLLMs generally fall into two architectural categories:
(1) Decoder-only (DecOnly) architectures, such as LLaVA-OneVision~\citep{li2024llava}, Qwen-VL~\citep{wang2024qwen2}, DeepSeek’s Janus~\citep{chen2025janus}, and InternVL~\citep{chen2024internvl}; and
(2) Encoder-decoder (EncDec) architectures, including LLaMA-3.2 Vision~\citep{chi2024llama}, NVLM-X~\citep{dai2024nvlm}, and Flamingo~\citep{alayrac2022flamingo}.
The key difference lies in how vision and text tokens are processed. Decoder-only models concatenate vision and text tokens and feed them together into the language model; vision tokens participate in every generation step. In contrast, encoder-decoder models use cross-attention modules to align multimodal inputs: vision tokens interact with text tokens only via these cross-attention layers, which are inserted between self-attention layers.

\begin{table}[t]
    \centering
%    \vspace{-10pt}
    \caption{Model configurations for four representative MLLMs with input image of 904$\times$904 pixels.}
    \label{tab1}
    \footnotesize
    \begin{tabularx}{\linewidth}{l >{\centering\arraybackslash}X >{\centering\arraybackslash}X >{\centering\arraybackslash}X >{\centering\arraybackslash}X}
    \toprule
    \textbf{MLLM Model Name} & \textbf{Architecture}  & \textbf{Image Encoder} & \textbf{Total Image} & \textbf{LLM Backend} \\
     & &  \textbf{\textit{(\#Params)}} &\textbf{Token Size}& \textbf{\textit{(\#Params)}} \\
    \midrule
    Llama3.2-Vision 11B & Encoder-Decoder  & ViT-H/14 (630M) & $6516$ & Llama 3.1 (8B) \\
    Llama3.2-Vision 90B & Encoder-Decoder  & ViT-H/14 (630M) & $6516$ & Llama 3.1 (70B) \\
    Qwen2.5-VL 7B & Decoder-only  & ViT (670M) & $7410$ & Qwen2.5 (7B) \\
    Qwen2.5-VL 72B & Decoder-only & ViT (670M) & $7410$ & Qwen2.5 (72B) \\
    \bottomrule
    \end{tabularx}
    \vspace{-5mm}
\end{table}

\textbf{Additional Overheads.} With the continuous increase in model scale and complexity, the computational cost of inference also grows, especially for MLLMs. Compared to language models, MLLMs introduce overhead from two main sources (shown in Fig.~\ref{fig1}a): (1) Increased Architectural Complexity: extra components such as vision encoders and cross-attention layers make the model heavier. (2) Extended Context Length: Multimodal data, once encoded into tokens, are concatenated with text prompts, increasing input length during inference~\citep{radford2021learning}. Specifically, this extension increases computational load, memory usage, and causes degradation in latency and throughput. Table~\ref{tab1} illustrates the encoder sizes and vision token lengths added by four mainstream open-source MLLMs.

\subsection{LLM Serving Systems}
\textbf{Existing LLM Serving.}
Existing serving systems~\citep{rasley2020deepspeed, aminabadi2022deepspeed, 2023lmdeploy}, such as vLLM~\citep{kwon2023efficient} and SGLang~\citep{zheng2024sglang}, have introduced a range of techniques to accelerate inference. To avoid redundant computation, these systems cache the key-value (KV) states of tokens for reuse. This design splits the inference process into two stages: the prefill stage and the decode stage. The prefill stage computes the KV cache for all input tokens and generates the first output token, while the decode stage generates one token per iteration. Consequently, the prefill stage carries a significantly heavier computational load. To this end, ORCA~\citep{yu2022orca} introduced continuous batching, enabling the system to process the prefill and decode stages of multiple requests in an uninterrupted stream of batches. To alleviate the burden of long input contexts, chunked prefill~\citep{agrawal2023sarathi, holmes2024deepspeed} splits long contexts into smaller segments that can be interleaved with decoding. However, interference between the two stages remains difficult to eliminate. A promising solution is prefill-decoding disaggregation, which places the two stages on separate GPU instances~\citep{zhong2024distserve, patel2024splitwise, jin2024p}. LoongServe~\citep{wu2024loongserve} further introduces elastic sequence parallelism to further optimize parallelism and resource allocation under this disaggregated architecture.

\textbf{Coupled Multimodal Serving.}
Despite these advancements, SOTA serving systems exhibit significant coupling issues when deployed for MLLMs. This coupling occurs at both the service and infrastructure levels. At the service level, existing systems treat unimodal (text-only) and multimodal requests identically, routing both through the same inference pipeline, despite their vastly different computational requirements. At the infrastructure level, all components—including preprocessors, vision encoders, and LLM backends—are colocated on the same hardware server. These tightly coupled components must share compute and memory resources while being constrained to a uniform batching and model parallelism strategy, leading to resource contention and inefficiency.

\subsection{Research Challenges and Motivations}
\label{sec:problem}

Existing tightly coupled systems face significant limitations when serving multimodal models.
\textit{\textbf{Service-level Problem:}} Multimodal and text-only requests differ substantially in resource demands due to longer context lengths and the injection of multimodal data (as illustrated in Fig. \ref{fig1}c). Meeting the same SLO requires different resource allocations for each request type. However, coupling their inference execution reduces overall resource utilization and increases the risk of SLO violations. It also prevents leveraging their distinct characteristics to enable more efficient inference and resource scheduling. \textit{\textbf{Architectural Problem:}} Encoder-decoder models\citep{alayrac2022flamingo,dai2024nvlm,chi2024llama}, which incorporate cross-attention mechanisms, are ill-suited for mixed-batch processing. Combining both request types in a single batch not only increases the latency of text-only requests but also degrades overall throughput.

% \begin{wrapfigure}{r}{0.6\linewidth}
%     \vspace{-4mm}
%     \centering
%     \includegraphics[width=\linewidth]{workload_chart1.pdf}
%     \caption{Production trace for MLLM Serving during a day.}
%     \label{fig2}
%     \vspace{-10pt}
% \end{wrapfigure}

These lead to our \textit{\textbf{Key Insight 1:}} To better meet the distinct requirements of each request type, a \textit{modality-aware decoupled} inference architecture should be adopted, where text-only and multimodal requests are processed independently.

However, naive static decoupling with fixed resource allocation is ineffective under dynamic workloads. It cannot adapt to changing request distributions (e.g., sudden spikes in image traffic) or adjust parallelism strategies based on evolving resource demands. 

This leads to our \textit{\textbf{Key Insight 2:}} An \textit{elastic} serving system is essential for handling dynamic multimodal workloads. Such a system must extend static architectures to support dynamic resource reallocation and stage-specific parallelism adjustments. By enabling scalable execution across elastic instances, it effectively alleviates compute and memory bottlenecks in multimodal serving.

\section{Elastic Multimodal Parallelism}

% \begin{wrapfigure}{r}{0.5\linewidth}  % r表示图片在右边，l表示左边
%     \vspace{-10pt}
%     \centering
%     \includegraphics[width=\linewidth]{framework1.pdf}
%     \caption{Overview of ElasticMM.}
%     \vspace{-10pt}
%     \label{fig3}
% \end{wrapfigure}
\begin{figure*}
   \centering
   \includegraphics[width=0.85\linewidth]{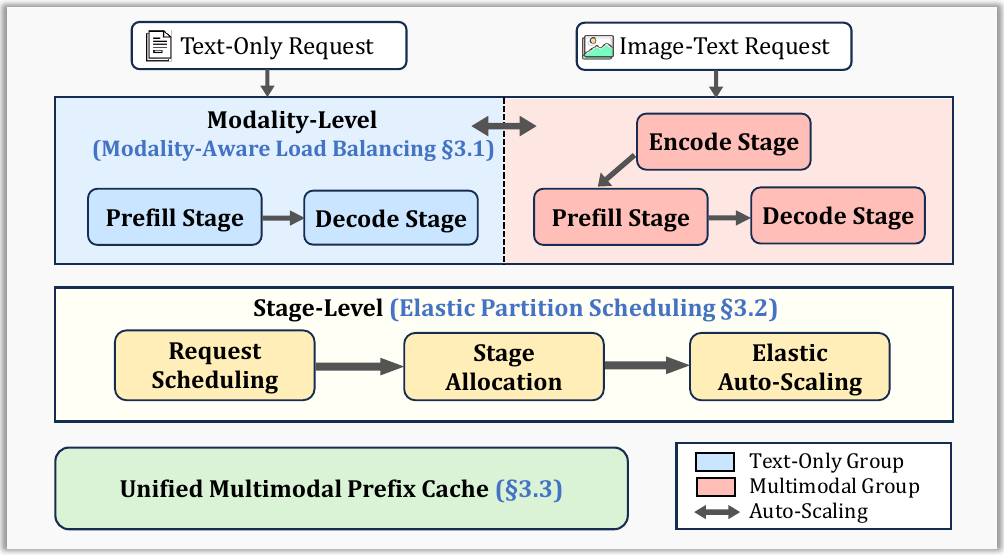}
   \caption{Framework diagram of ElasticMM. The figure illustrates a two-level scheduling framework that collaboratively enables elastic multimodal parallelism.}
   \label{fig3}
   \vspace{-5pt}
\end{figure*}

As discussed in the previous section, it is essential to build an elastic MLLM serving system that can dynamically adjust both resource allocation and parallelism strategies. To this end, we introduce a hierarchical framework that enables elastic scheduling, where all instances are organized into two levels, as illustrated in Fig.~\ref{fig3}.

At the first level, the \textbf{modality level}, instances are grouped based on the modality of the models they serve (e.g., text-only or multimodal). At the second level, the \textbf{stage level}, the inference pipeline is further disaggregated into distinct stages, such as encoding, prefill, and decode. At each level, our system provides both \textit{\textbf{decoupling}} and \textit{\textbf{elasticity}}, which are the two core strengths of ElasticMM. This design eliminates resource contention and maximizes utilization.

At the modality level, we design a \textit{modality-aware load balancing strategy} to dynamically allocate resources across modality groups in response to fluctuating request loads. At the stage level, we apply \textit{elastic partition scheduling} to enable flexible parallelism tailored to the needs of each inference stage.

\subsection{Modality-Aware Load Balancing}
\label{sec3.1}
To tackle the inter-modality load imbalance outlined in Section~\ref{sec:problem}, we combine both proactive and reactive mechanisms. Specifically, following prior approaches~\citep{zhang2023shepherd,li2023alpaserve,wu2024dlora}, we apply a proactive mechanism to allocate resources to different modality groups. However, proactive mechanism alone cannot effectively handle sudden request surges (e.g., bursty multi-image streams). We therefore design a reactive scaling mechanism to dynamically expand capacity upon detecting resource shortages.

\textbf{Proactive Mechanism.}
Our analysis of load patterns in inference services reveals twofold: although the short-term arrival pattern of a single request stream is difficult to predict, aggregating multiple streams results in smooth and periodic patterns (e.g., lower load at night, higher during the day). As observed in previous work~\citep{qiumodserve}, text-only requests exhibit stable loads, while multimodal traffic is marked by pronounced surges.
Based on the predictability of long-term workload, idle elastic instances can be proactively assigned to modality groups. Following previous work~\citep{zhang2023shepherd}, our goal in allocation is to maximize the minimum burst tolerance across modality groups. Burst tolerance ($bt$) is defined as the ratio between peak-available and average-required instances per group in Equation~\ref{eq1}. To achieve this, we employ a fast and effective greedy strategy: incrementally assigns each instance to the group currently lowest burst tolerance, continuing until resources are fully allocated.
{
\small
\begin{align}
bt(i) &= \frac{\text{\# Instances } i \text{ can use for its peak load }}{\text{\# Instances } i \text{ can use for its average load } } = \frac{N_i^{\text{peak}}}{N_i^{\text{avg}}}
\label{eq1}
\end{align}
}

\textbf{Reactive Scaling.}
Due to unpredictable short-term workloads, the system may face sudden traffic bursts, such as long-text or image-heavy requests in real-time scenarios. The system evaluates the trade-off between adjusting intra-group parallelism and triggering inter-group reactive scaling, based on a gain-cost model described in Section~\ref{sec3.2}. If inter-group reactive scaling is more beneficial, the modality-level manager selects instances to preempt from other groups with minimal impact. When an instance $E$ is preempted, its workload is merged into other instances at the same inference stage.

\subsection{Elastic Partition Scheduling}
\label{sec3.2}

After addressing inter-group GPU allocation, we turn to intra-group request scheduling and parallelism adjustment, adapting to the distinct characteristics of each inference stage. For example, prefill is compute-bound; decode is always memory-bound and scales poorly~\citep{zhong2024distserve,patel2024splitwise,jin2024p}. For stages with sublinear scalability, allocating more idle GPUs improves performance. Conversely, for stages with poor scalability, it is more efficient to limit the number of GPUs involved~\citep{wu2024loongserve}. Fig.~\ref{fig4} illustrates request assignments and parallel execution modes in the ElasticMM service, showing that instance counts vary by modality group and inference stage. We propose Elastic Partition Scheduling, orchestrated by the stage-level manager. This strategy decomposes the scheduling challenge into three subproblems: (1) \textit{\textbf{Request scheduling}}, (2) \textit{\textbf{Stage Allocation}}, and (3) \textit{\textbf{Elastic auto-scaling}}.

\begin{figure*}
   \centering
   \includegraphics[width=0.85\linewidth]{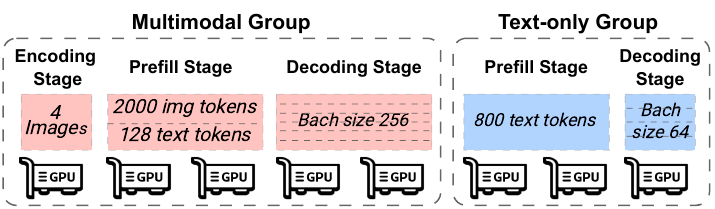}
   \caption{Illustration of the elastic scheduling space in EMP.}
   \label{fig4}
   \vspace{-5pt}
\end{figure*}

\textbf{Request Dispatching.}
This step selects a subset of pending requests $R_p \subseteq P$ from the queue $P$ a First-Come-First-Served (FCFS) policy for the prefill phase~\citep{kwon2023efficient,yu2022orca}. One exception: if a text-only dialogue is redirected to a multimodal group due to associated multimodal requests, it’s prioritized—this helps overlap migration overhead and frees KV slots earlier. The dispatch process must address two types of constraints: GPU memory and compute throughput. For memory constraints, the manager only adds requests to $R_p$ if sufficient KV slots are available~\citep{wu2023fast}. There's a tipping point where the system shifts from memory-bound to compute-bound. Before this point, adding requests to $R_p$ improves utilization; after that, additional requests degrade performance due to extended execution time. We estimate this point by analyzing the upper bound of prefill time under memory bound.

\begin{figure*}
   \centering
   \includegraphics[width=0.8\linewidth]{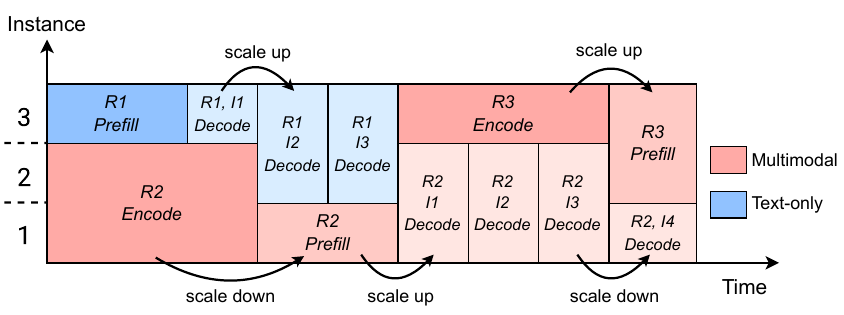}
   \caption{Example of elastic auto-scaling in three instance.}
   \label{fig5}
   \vspace{-5pt}
\end{figure*}

\textbf{Stage Allocation.}
Once the prefill request set $R_p$ is formed, the system allocates a set of elastic instances  $E_p$ to serve it, aiming to maximize GPU efficiency. Inspired by Loongserve~\citep{wu2024loongserve}, the system will prioritize allocating idle instances to $R_p$ for prefill stage. If the available KV cache slots in idle instances are insufficient, $R_p$ is permitted to preempt instances set $E_d$ in the decoding stage $B_d$. The instance with maximum unused slots are denoted as $e_{max}$. Meanwhile, drawing insights from existing research, compute-intensive stages like prefill can still benefit from increased parallelism. Therefore, the stage-level manager continuously considers further preemptions by selecting the $e_{max}$, migrating its KV cache to other decode instances. The gain-cost model is formulated as follows:
{\small
\begin{align}
\text{Gain} &= \sum_{r \in R_p} \frac{T(R_p, E_p) - T(R_p, E_p \cup e_{max})}{r.\text{input\_len}} &
\text{Cost} &= \sum_{r \in B_d} \frac{M(e_{max})+ w \cdot L(B_d,E_d - e_{max})}{r.\text{output\_len}}
\end{align}
}The gain of preemption is quantified as the acceleration gained by adding $e_{max}$ to the prefill process. The cost accounts for both migration overhead $M(e_{max})$ and the performance impact $L$ on the preempted computation. A tunable penalty factor $w$ is introduced to control the aggressiveness of preemption, enabling flexible system behavior under varying workloads. 
Within a single inference stage, we prioritize Data Parallelism (DP). Alternatives such as Tensor Parallelism are only employed when a single GPU cannot hold the model weights. This design offers several advantages: 1) During elastic scaling, only the KV cache needs to be migrated, avoiding expensive weight transfers; 2) DP enables more flexible utilization of fragmented GPU resources, whereas tensor parallelism typically requires an even number of devices and stricter placement constraints.

\textbf{Elasic Auto-Scaling.}
 ElasticMM supports elastic scaling, including at the stage level. Since resource allocation is primarily guided by prefill, we monitor the decode phase to trigger elastic adjustments. Given that the decode stage exhibits poor scalability, we shrink it to the minimum parallelism. When the GPU resource becomes insufficient during decoding, the system triggers reactive scaling. Based on prior observations, decode bottlenecks often occur at FFN layers and depend on batch size. We establish scaling thresholds through offline profiling. The scaling-up first tries to allocate idle instances from the same group. If that fails, it first selects a candidate instance $e_{\text{max}}$ from intra-group prefill instances and a candidate $e'_{\text{max}}$ from inter-group instances, using the following gain-cost model to estimate the trade-off between potential speedup and migration overhead. The instance with the highest net gain is selected for preemption. If inter-group preemption offers greater benefit, the system triggers the reactive scaling mechanism described in Section~\ref{sec3.1}.
 {\small
\begin{align}
\text{Gain} &= \sum_{r \in B_{d}} \frac{\text{AvgLat}_d - T(B_d, E_d \cup e_{max})}{r.\text{output\_len}} & 
\text{Cost} &= \sum_{r \in R'_{p}} \frac{M(e_{max})+ w \cdot L(R'_{p},E'_p-e_{max})}{r.\text{input\_len}}
\end{align}
}Fig.~\ref{fig5} visualizes auto-scaling. Given the higher computational complexity of multimodal encoding compared to text-only prefill and decode, $R2$'s encode stage is initially allocated more resources. During $R1$'s decoding progression, the growing token sequence increases both computational load and KV cache memory pressure. This necessitates expanding decode stage ($(R1, I2)$ and $(R1, I3)$). Later, when multimodal request $R3$ arrives, the multimodal group experiences resource contention. It preempts idle resources from the text-only group for encoding. After $R3$'s encoding, allocating more resources to its prefill stage yields better throughput, so $R3$ automatically scales up to more instances.

\subsection{Multimodal Inference Optimization}
\textbf{Unified Multimodal Prefix Cache.}
In real-world service scenarios, user requests often exhibit redundancy. For example, requests may share identical system prompts in their headers or involve repeated transmission of the same images. To address the issue of redundant computation and data transmission in multimodal inference systems, we propose a Unified Multimodal Prefix Cache optimization strategy. This strategy integrates text prefix caching with multimodal input caching to build a unified caching scheme. We categorize cached objects into two pools: one for tokens encoded from multimodal inputs and the other for prefix tokens from unified sequences (including both multimodal and text tokens). Following prior works~\citep{zheng2024sglang,juravsky2024hydragen}, each cache pool is managed with a Least Recently Used (LRU) dynamic eviction strategy to keep memory usage under control.
When a multimodal input is received, we generate a hash. If the hash matches an existing entry, we skip re-encoding and use the cached tokens. After merging with the text tokens, we check a prefix tree in the second cache pool to find the longest matching prefix. That portion of the sequence skips the prefill step and directly uses the cached key-value results. This unified caching mechanism significantly reduces redundant visual model overhead as well as repetitive computation in the language model.

\textbf{Non-blocking Encoding.}
As illustrated in Fig.~\ref{fig1}a, the vision encoder of a multimodal large model introduces substantial overhead due to image preprocessing and encoding—often taking more than five times longer than the prefill stage. However, existing inference frameworks tightly couple the encoding and prefill stages, leading to blocking behavior during encoding that delays subsequent stages. This adds latency to the first token response (TTFT) and reduces overall throughput. Our solution decouples the vision and language models by isolating image preprocessing and encoding into a separate process or even a separate instance—executed asynchronously. 
By incorporating Unified Multimodal Prefix Cache and Non-blocking Encoding, ElasticMM effectively reduces redundant computations and inter-stage interference, thereby improving the efficiency of multimodal inference.

\section{Evaluation}
\subsection{Experimental Setup}
\label{sec4.1}
\textbf{Model and Dataset.}
We select LLaMA3.2-Vision-11B~\citep{chi2024llama} and Qwen2.5-VL-7B~\citep{wang2024qwen2} to represent encoder-decoder and decoder-only architectures, respectively. Two open-source multimodal datasets are used, each containing a mix of multimodal and text-only requests: VisualWebInstruct\citep{jia2025visualwebinstruct} is a large-scale dataset collected from over 700K unique web URLs; ShareGPT-4o\citep{chen2024far} comprises 50K images of varying resolutions along with corresponding text prompts sourced from the multimodal GPT-4o model.
VisualWebInstruct contains longer average text inputs, while ShareGPT-4o includes higher-resolution images, making the two datasets complementary. This combination enhances the comprehensiveness and rigor of our evaluation. Following prior work~\citep{kwon2023efficient,wu2024loongserve}, we use a Poisson distribution to generate variable request arrival rates (requests per second, QPS) and incorporate real-world production service traces to simulate realistic workload distributions.

\textbf{Testbed.}
We evaluate ElasticMM on a high-end workstation equipped with eight NVIDIA A800 80GB GPUs, two 64-core Intel Xeon 8358P CPUs, and 2 TB of DDR4 memory. The NVLink bandwidth between any two GPUs is 400 GB/s. We leave multi-node distributed studies to future work, as this testbed already demonstrates superior performance.

\textbf{Baselines.}
We compare ElasticMM with two baselines. The first baseline uses a coupled system for MLLM serving, with the SOTA system vLLM~\citep{kwon2023efficient} (v0.6.6) as representative. It follows the architecture of this version, with additional model code to support Qwen2.5-VL. In this system, all inference stages are coupled and executed on the same hardware. The second baseline is DistServe~\citep{zhong2024distserve}, a decoupled architecture that colocates the encode and prefill stages while separating them from decode, using static resource allocation. We extend this system to support multimodal inference. To evaluate the effectiveness of our proposed techniques, we build ElasticMM on top of vLLM and construct several variants by selectively enabling each technique for ablation studies, allowing us to isolate their individual contributions. Further implementation details are provided in Appendix~\ref{implementation}.

\textbf{Metrics.}
We evaluate service quality using latency and throughput metrics. For each request rate, we measure the \textit{normalized input latency} (i.e., average prefill time divided by input length) and \textit{normalized output latency} (i.e., average decoding time divided by output length). To compare system throughput under consistent conditions, we define a uniform SLO and record the maximum throughput achievable within that target. Following prior work~\citep{wu2023fast,wu2024loongserve}, we set the SLO to 10$\times$ the latency under light load and then scale it with a constant factor (ranging from one to five) to evaluate performance under both relaxed and strict conditions.

\begin{figure}
\centering
    \setlength{\fboxsep}{0pt}      % 取消 minipage 内边距
    \setlength{\tabcolsep}{0pt}    % 取消表格列间距（如果嵌套在表格中）
    \setlength{\parindent}{0pt}    % 取消段落缩进
    \offinterlineskip               % 取消行间间距（关键！）
   \includegraphics[width=1.0\linewidth]{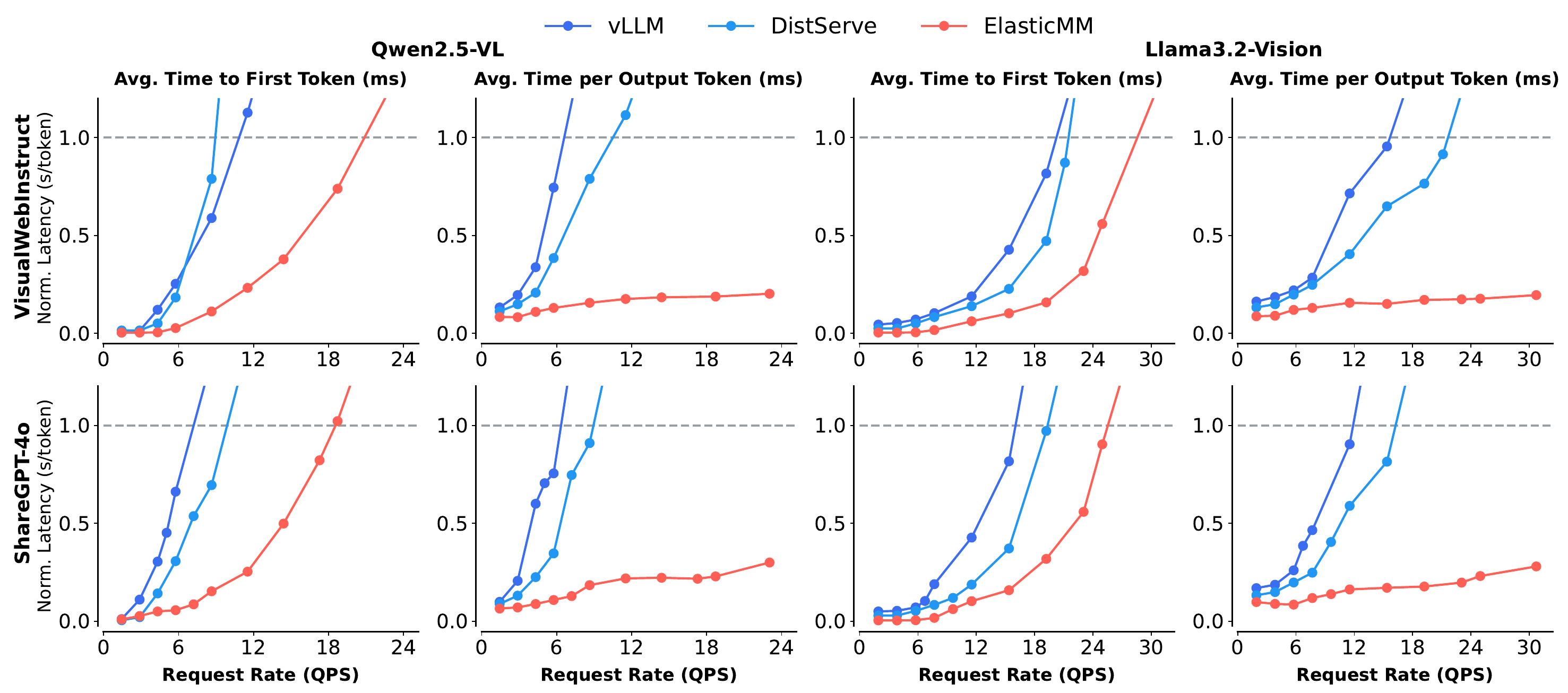}
   \caption{The average input and output latency of ElasticMM and baseline MLLM serving systems with the Qwen2.5-VL-7B and Llama3.2-Vision-11B under two real-world workloads. ElasticMM consistently demonstrates the lowest latency across all cases.}
   \vspace{-5pt}
   \label{fig6}
\end{figure}

\subsection{End-to-End Performance}
\textbf{Input Latency.} In our evaluation on two open-source large multimodal models, we compare ElasticMM against two baselines using two real-world workloads, focusing on input and output latency metrics. Fig.~\ref{fig6} presents the performance results, where higher request rates correspond to heavier workloads, and lower latency are better. ElasticMM benefits from its decoupled design, the prefill stage proceeds without interference from image encoding. As a result, ElasticMM significantly reduces \textit{input token latency}, i.e., time-to-first-token (TTFT). Although the DistServe shows some improvements, its static resource allocation cannot dynamically adjust under unbalanced loads, leading to underutilized resources and a rapid increase in latency at higher request rates. Fig.~\ref{fig6} demonstrates that ElasticMM consistently achieves the lowest input latency across all load levels. On the ShareGPT-4o dataset, ElasticMM reduces TTFT by up to 4.2$\times$ and 3.5$\times$ for Qwen2.5-VL (decoder-only) and LLaMA3.2-Vision (encoder-decoder), respectively. On VisualWebInstruct, TTFT is reduced by up to 3.7$\times$ and 2.9$\times$. The more substantial gains on decoder-only models can be attributed to their heavier prefill computation, which exacerbates conflicts with image encoding. Moreover, the stronger performance gains on the more visually intensive ShareGPT-4o dataset further validate the effectiveness of our multimodal inference optimizations.

\textbf{Output Latency.} Since ElasticMM decouples computation stages and elastically allocates them across separate instances, the decoding phase is well-isolated from the interference of encoding or prefill stages. This results in consistently lower output latency, outperforming all baselines. In comparison, vLLM executes all stages on the same instance, leading to severe resource contention as request rates increase, where interference from the encode and prefill stages severely degrades the performance of the decode stage and results in a marked increase in output latency. While DistServe moderately alleviates latency by separating the prefill and decode stages, its static resource allocation causes memory contention on decode nodes at higher request rates, leading to a significant increase in output latency.

\begin{figure}
\centering
    \includegraphics[width=\linewidth]{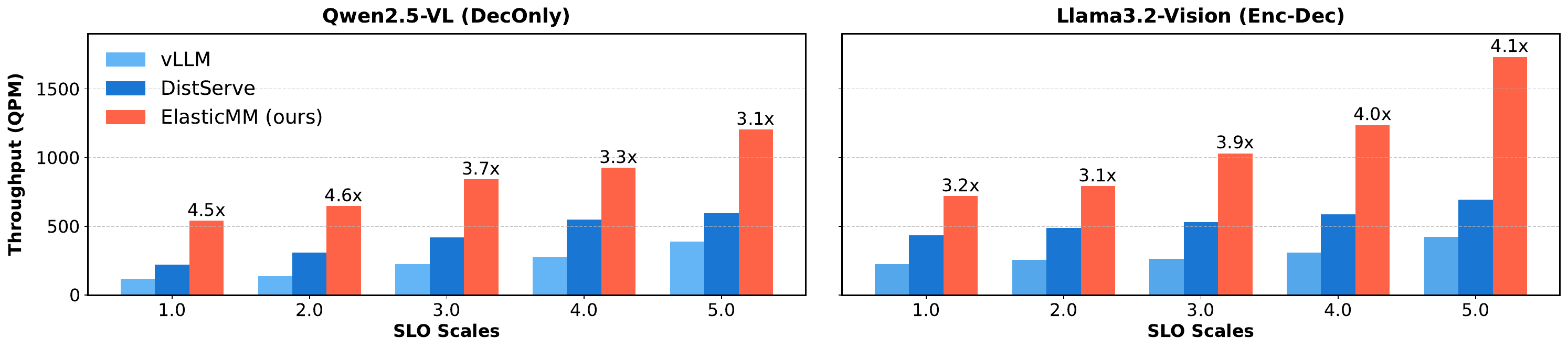}
    \caption{Maximum throughput meeting SLO.}
   \vspace{-5pt}
   \label{fig7}
\end{figure}

\begin{figure}[htbp]

    \centering
    % 取消所有可能影响间距的参数
    \setlength{\fboxsep}{0pt}      % 取消 minipage 内边距
    \setlength{\tabcolsep}{0pt}    % 取消表格列间距（如果嵌套在表格中）
    \setlength{\parindent}{0pt}    % 取消段落缩进
    \offinterlineskip               % 取消行间间距（关键！）
    
              % 注释掉换行符
    \begin{minipage}{0.5\textwidth}
        \includegraphics[width=\linewidth]{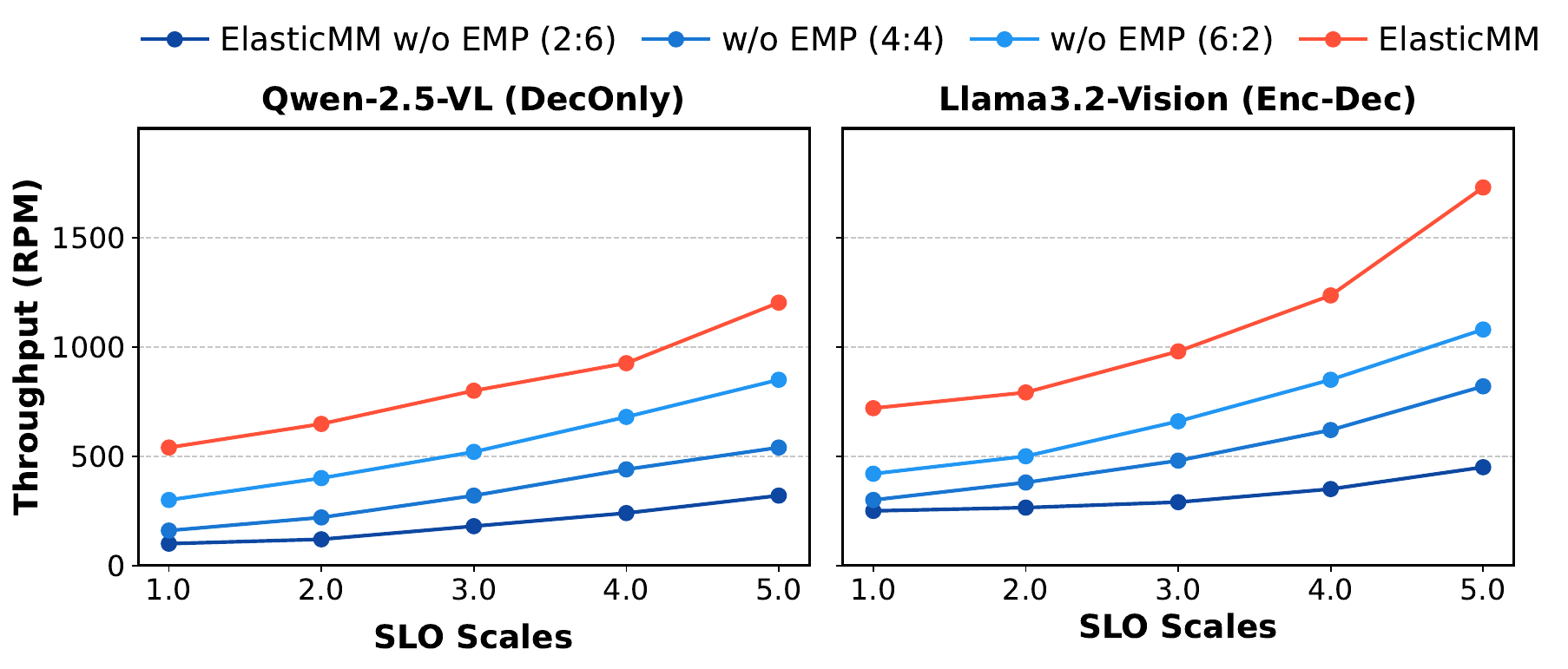}
        \caption{Throughput impact of resource allocation}
        \label{fig8}
    \end{minipage}% 
    \begin{minipage}{0.5\textwidth}
        \includegraphics[width=\linewidth]{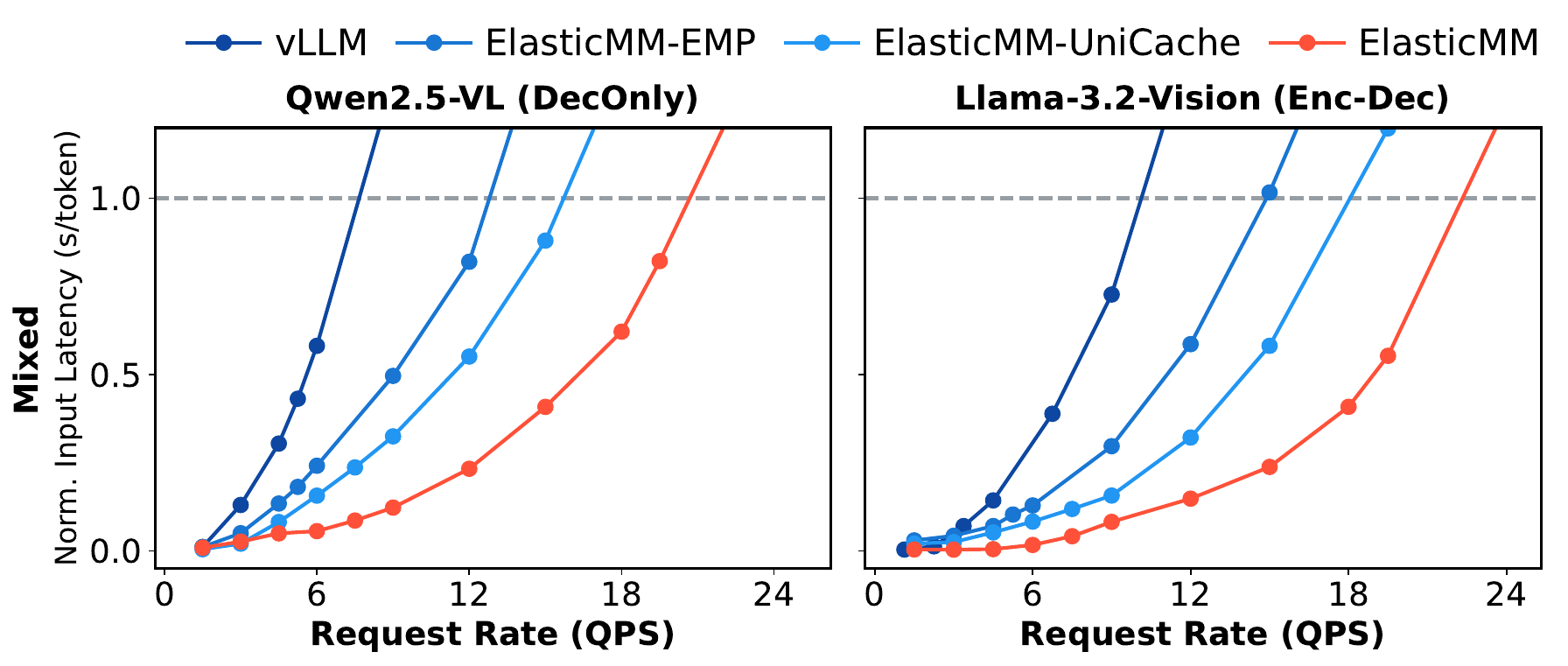}
        \caption{Ablation studies of our optimizations.}
        \label{fig9}
    \end{minipage}%    
    \vspace{-5pt}
\end{figure}

\textbf{Throughput.} Fig.~\ref{fig7} further evaluates the maximum throughput under linearly scaled service-level objectives (SLOs), ranging from 1$\times$ to 5$\times$, simulating both strict and relaxed service conditions. ElasticMM achieves the highest throughput across all SLO settings. Specifically, on the ShareGPT-4o dataset, it delivers up to 4.5$\times$ and 3.2$\times$ higher throughput than vLLM on Qwen2.5-VL and LLaMA3.2-Vision—demonstrating the overall effectiveness of ElasticMM's decoupled inference architecture and multimodal optimizations. Furthermore, ElasticMM achieves up to 2.3$\times$ higher throughput than DistServe, benefiting from its hierarchical elastic scheduling strategy and non-blocking encoding design.

\subsection{Ablation Study}

% \begin{wrapfigure}{r}{0.5\linewidth}  % r表示图片在右边，l表示左边
%     \centering
%     \vspace{-12pt}
%     % 取消所有可能影响间距的参数
%     \setlength{\fboxsep}{0pt}      % 取消 minipage 内边距
%     \setlength{\tabcolsep}{0pt}    % 取消表格列间距（如果嵌套在表格中）
%     \setlength{\parindent}{0pt}    % 取消段落缩进
%     \offinterlineskip               % 取消行间间距（关键！）

%     \includegraphics[width=\linewidth]{ttft_comparison.png}
%     \caption{Ablation studies of our optimizations.}
%     \label{fig9}
%     \vspace{-6pt}
% \end{wrapfigure}
\textbf{Effectiveness of Elastic Multimodal Parallelism.} To evaluate the effectiveness of Elastic Multimodal Parallelism (EMP), we conduct an ablation study on the achievable P90 effective throughput under varying scaled SLOs. The baselines include three static resource allocation strategies: (1) a text-dominant policy, (2) an equal allocation policy, and (3) a multimodal-dominant policy. These represent ElasticMM variants without EMP enabled, with resource ratios defined by the number of multimodal versus text instances. To better demonstrate EMP's strength under heavy load, we use the ShareGPT-4o dataset, which contains images with higher average resolution. The results are shown in Fig.~\ref{fig8}. As illustrated, static resource allocation leads to suboptimal performance—even when equipped with the two key multimodal inference optimizations. Whether biased toward a specific modality group or evenly distributed, static schemes cannot efficiently handle dynamically shifting inference workloads. In contrast, ElasticMM dynamically adjusts inter-group resource allocation via load balancing and applies elastic partitioned scheduling to enable fine-grained parallelism across different inference stages. As a result, ElasticMM achieves substantial throughput improvements on both representative models: 1.8$\times$ for Qwen-2.5-VL and 2.3$\times$ for Llama-3.2-Vision.

\textbf{Effectiveness of MLLM Inference Optimizations.} To evaluate the effectiveness of the two optimization techniques—Unified Multimodal Prefix Cache and Non-blocking Encoding—in reducing input token latency (i.e., TTFT), we conduct an ablation study on top of ElasticMM with EMP enabled. We incrementally apply the two optimizations to assess their individual and combined benefits. The baselines include: the system without either optimization (ElasticMM-EMP), the system with only Unified Multimodal Prefix Cache (ElasticMM-UniCache), and the fully optimized system (ElasticMM). To demonstrate the robustness of these optimizations, we generate requests by sampling from a mixed dataset composed of two distinct sources. As shown in Fig.~\ref{fig9}, applying EMP alone provides limited improvements in token input latency. The Unified Multimodal Prefix Cache significantly reduces redundant computation and data transfer in both the vision and language models, thus effectively lowering latency. The Non-blocking Encoding technique eliminates the blocking effect between the image encoding stage and subsequent inference stages, enabling a more efficient inference pipeline and leading to further latency reduction. The normalized token input latency results confirm that both optimizations provide consistent performance gains across most requests.

\section{Related Work}
\textbf{MLLM Serving Optimizations.} 
Recent work on MLLM serving focuses on improving inference computation and memory efficiency. For example, MobileVLM~\citep{wu2024mobilevlm}, TinyGPT-V~\citep{yuan2023tinygpt}, and TinyLLaVA~\citep{zhou2024tinyllava} reduce the backbone model size to improve speed. Other approaches, such as Dynamic-LLaVA~\citep{huang2024dynamic}, LLaVA-Mini~\citep{zhang2025llava}, and VTM~\citep{lin2025boosting}, introduce visual token pruning or compression to reduce context length. InfMLLM~\citep{ning2024inf} and Elastic Cache~\citep{liu2024efficient} reduce memory usage by pruning key-value (KV) caches based on token importance. These methods operate at the model level, \textit{trading off accuracy for computational efficiency}, which can lead to degraded performance on vision tasks and may be constrained by specific model architectures. In contrast, our approach introduces \textit{no accuracy degradation}, and is compatible with all MLLM architectures. Therefore, we do not compare against these optimization methods in this work. Recent work~\citep{qiumodserve,singh2024efficiently} attempts to separate imultimodal model inference stages, but both types of requests still remain batched together during LLM backend inference, leading to persistent efficiency issues with Decoder-Encoder architectures. In contrast, the key distinction of our work is the introduction of modality group isolation alongside stage separation, creating a two-tier separation architecture, upon which we develop an elastic scheduling mechanism across modality groups. 

\textbf{LLM Serving Optimizations.} Recent work on optimizing general LLM serving has explored disaggregated architectures, including Splitwise~\citep{patel2024splitwise}, DistServe~\citep{zhong2024distserve}, and Mooncake~\citep{qin2024mooncake}, but their static parallelism and partitioning strategies lack the flexibility to handle dynamic workloads. Other studies have focused on improving GPU operator efficiency, such as FlashAttention~\citep{dao2022flashattention} and Flash-Decoding~\citep{dao2023flashattention}. These methods are orthogonal to ElasticMM and can be integrated to further enhance the performance of its LLM backend. Additional optimizations target KV cache management~\citep{kwon2023efficient}, request scheduling~\citep{sun2024llumnix,qiu2024efficient,patke2024queue,su2025seesaw}, and batching efficiency~\citep{yu2022orca,agrawal2023sarathi,holmes2024deepspeed}. FlexPipe~\citep{zhao2025flexpipe} introduces a dynamic pipeline parallelism framework to improve training efficiency. While these systems achieve strong performance in traditional LLM serving, they rely on tightly coupled architectures that do not meet the unique requirements of MLLM workloads. Therefore, we implement ElasticMM on top of the SOTA LLM serving system vLLM, without incorporating these additional optimizations.

\section{Conclusion}
In this paper, we first analyze the limitations of existing systems, which exhibit tightly coupled designs when handling serving MLLMs, making them inefficient under multimodal workloads. We propose two key insights: an efficient MLLM serving system must be both decoupled and elastic.
Based on these principles, we introduce Elastic Multimodal Parallelism, a novel serving paradigm implemented in our system ElasticMM. ElasticMM incorporates three key innovations: (1) modality-aware load balancing, (2) elastic partition scheduling, and (3) multimodal inference optimizations, enabling dynamic resource adaptation across different request types and inference stages. Comprehensive evaluations on diverse real-world datasets show that ElasticMM achieves up to 4.2$\times$ reduction in TTFT and 3.2-4.5$\times$ higher throughput compared to vLLM while consistently meeting SLOs, establishing a new paradigm for scalable multimodal AI service infrastructure.

In future work, we plan to extend our evaluation to large-scale, multi-node clusters. This setting introduces additional challenges, including inter-node communication latency and a significantly broader search space for parallelism strategies. We leave the exploration of these challenges to future work, as we believe they open up important avenues for advancing scalable and efficient multimodal serving systems.

\section*{Acknowledgments}
We thank the anonymous reviewers for their insightful comments and feedback. 
This work was supported in part by the National Key Research and Development Program of China (Grant No. 2025YFB30037002), 
the National Natural Science Foundation of China (Grant Nos. 62032023 and T2125013), 
and the Innovation Funding of ICT, CAS (Grant No. E461050). Yang You's research group is being sponsored by NUS startup grant (Presidential Young Professorship), Singapore MOE Tier-1 grant, ARCTIC grant.

{
\small
\bibliographystyle{IEEEtran}
\bibliography{neurips_2025}
}

%%%%%%%%%%%%%%%%%%%%%%%%%%%%%%%%%%%%%%%%%%%%%%%%%%%%%%%%%%%%

\newpage
\section*{NeurIPS Paper Checklist}

\begin{enumerate}

\item {\bf Claims}
    \item[] Question: Do the main claims made in the abstract and introduction accurately reflect the paper's contributions and scope?
    \item[] Answer: \answerYes{}
    \item[] Justification: The abstract and introduction accurately summarize our key contributions, methodology, and experimental results. Section \ref{sec1} explicitly outlines our claims and their scope.

\item {\bf Limitations}
    \item[] Question: Does the paper discuss the limitations of the work performed by the authors?
    \item[] Answer: \answerYes{} % Replace by \answerYes{}, \answerNo{}, or \answerNA{}.
    \item[] Justification: In Section \ref{sec4.1}, we explicitly discuss the limitations of our system, noting that the current implementation has only been tested in a single-node environment, and we identify multi-node deployment and testing on larger clusters as important directions for future work.

\item {\bf Theory assumptions and proofs}
    \item[] Question: For each theoretical result, does the paper provide the full set of assumptions and a complete (and correct) proof?
    \item[] Answer: \answerYes{} % Replace by \answerYes{}, \answerNo{}, or \answerNA{}.
    \item[] Justification: Our experiments use multiple mainstream open-source models and datasets, making them inherently reproducible. We will soon release our code for full reproduction.

    \item {\bf Experimental result reproducibility}
    \item[] Question: Does the paper fully disclose all the information needed to reproduce the main experimental results of the paper to the extent that it affects the main claims and/or conclusions of the paper (regardless of whether the code and data are provided or not)?
    \item[] Answer: \answerYes{} % Replace by \answerYes{}, \answerNo{}, or \answerNA{}.
    \item[] Justification: Our experiments use multiple mainstream open-source models and datasets, making them inherently reproducible. We will soon release our code for full reproduction.

\item {\bf Open access to data and code}
    \item[] Question: Does the paper provide open access to the data and code, with sufficient instructions to faithfully reproduce the main experimental results, as described in supplemental material?
    \item[] Answer: \answerYes{} % Replace by \answerYes{}, \answerNo{}, or \answerNA{}.
    \item[] Justification: We have provided links to all open-source data and models used in our experiments, and our code will be made publicly available shortly.

\item {\bf Experimental setting/details}
    \item[] Question: Does the paper specify all the training and test details (e.g., data splits, hyperparameters, how they were chosen, type of optimizer, etc.) necessary to understand the results?
    \item[] Answer: \answerYes{} % Replace by \answerYes{}, \answerNo{}, or \answerNA{}.
    \item[] Justification: Core experimental settings are described in section \ref{sec4.1}, with additional details provided in the appendix \ref{implementation}.

\item {\bf Experiment statistical significance}
    \item[] Question: Does the paper report error bars suitably and correctly defined or other appropriate information about the statistical significance of the experiments?
    \item[] Answer: \answerYes{} % Replace by \answerYes{}, \answerNo{}, or \answerNA{}.
    \item[] Justification: We provide detailed reporting of error bars and their calculation methods in the supplementary materials.

\item {\bf Experiments compute resources}
    \item[] Question: For each experiment, does the paper provide sufficient information on the computer resources (type of compute workers, memory, time of execution) needed to reproduce the experiments?
    \item[] Answer: \answerYes{}% Replace by \answerYes{}, \answerNo{}, or \answerNA{}.
    \item[] Justification: Section \ref{sec4.1} provides detailed information about the computational resources used in our experiments.
    
\item {\bf Code of ethics}
    \item[] Question: Does the research conducted in the paper conform, in every respect, with the NeurIPS Code of Ethics \url{https://neurips.cc/public/EthicsGuidelines}?
    \item[] Answer: \answerYes{} % Replace by \answerYes{}, \answerNo{}, or \answerNA{}.
    \item[] Justification: Our research fully complies with the NeurIPS Code of Ethics.

\item {\bf Broader impacts}
    \item[] Question: Does the paper discuss both potential positive societal impacts and negative societal impacts of the work performed?
    \item[] Answer: \answerYes{} % Replace by \answerYes{}, \answerNo{}, or \answerNA{}.
    \item[] Justification: ElasticMM provides positive societal impact by introducing a new paradigm for multimodal large model serving and deployment, contributing to better and faster AI service implementation.
    
\item {\bf Safeguards}
    \item[] Question: Does the paper describe safeguards that have been put in place for responsible release of data or models that have a high risk for misuse (e.g., pretrained language models, image generators, or scraped datasets)?
    \item[] Answer: \answerNA{} % Replace by \answerYes{}, \answerNo{}, or \answerNA{}.
    \item[] Justification: This paper focuses on inference processes and does not include model training, so these risks are not applicable.

\item {\bf Licenses for existing assets}
    \item[] Question: Are the creators or original owners of assets (e.g., code, data, models), used in the paper, properly credited and are the license and terms of use explicitly mentioned and properly respected?
    \item[] Answer: \answerYes{} % Replace by \answerYes{}, \answerNo{}, or \answerNA{}.
    \item[] Justification: We respect the intellectual property rights of the original authors of the data and models used in this paper, citing them appropriately and using them in full compliance with their respective licenses.

\item {\bf New assets}
    \item[] Question: Are new assets introduced in the paper well documented and is the documentation provided alongside the assets?
    \item[] Answer: \answerNA{} % Replace by \answerYes{}, \answerNo{}, or \answerNA{}.
    \item[] Justification:  The paper does not release new assets.839

\item {\bf Crowdsourcing and research with human subjects}
    \item[] Question: For crowdsourcing experiments and research with human subjects, does the paper include the full text of instructions given to participants and screenshots, if applicable, as well as details about compensation (if any)? 
    \item[] Answer: \answerNA{} % Replace by \answerYes{}, \answerNo{}, or \answerNA{}.
    \item[] Justification:This paper does not involve crowdsourcing or research with human subjects.

\item {\bf Institutional review board (IRB) approvals or equivalent for research with human subjects}
    \item[] Question: Does the paper describe potential risks incurred by study participants, whether such risks were disclosed to the subjects, and whether Institutional Review Board (IRB) approvals (or an equivalent approval/review based on the requirements of your country or institution) were obtained?
    \item[] Answer: \answerNA{} % Replace by \answerYes{}, \answerNo{}, or \answerNA{}.
    \item[] Justification: This paper does not involve research with human subjects.

\item {\bf Declaration of LLM usage}
    \item[] Question: Does the paper describe the usage of LLMs if it is an important, original, or non-standard component of the core methods in this research? Note that if the LLM is used only for writing, editing, or formatting purposes and does not impact the core methodology, scientific rigorousness, or originality of the research, declaration is not required.
    %this research? 
    \item[] Answer: \answerNA{} % Replace by \answerYes{}, \answerNo{}, or \answerNA{}.
    \item[] Justification: The core method development in this research does not involve LLMs as any important, original, or non-standard components.

\end{enumerate}

%%%%%%%%%%%%%%%%%%%%%%%%%%%%%%%%%%%%%%%%%%%%%%%%%%%%%%%%%%%%

\appendix

\section*{Appendix A: Implementation Details}
\label{implementation}
We implement ElasticMM using 7,000 lines of code based on Python, C++. We build upon vLLM~\citep{kwon2023efficient} (v0.6.6), a stateof-the-art generative model inference platform. To ensure efficient GPU-to-GPU memory transfer of KV cache, we use PyTorch's distributed communication with the NCCL backend and GPU Direct RDMA. 

The frontend of ElasticMM uses the OpenAI API format, identical to vLLM, allowing users who have previously used vLLM to send requests to ElasticMM without any modifications. ElasticMM's Modality-level manager and Stage-level manager are primarily implemented in Python. However, some core logic, such as Stage Allocation, is implemented in C++ to improve the efficiency of loop functions. The Modality-level manager uses Ray for communication between elastic instances, while the Stage-Level manager assigns each batch to a Python coroutine to manage multiple batches simultaneously. Similar to vLLM, when tensor parallelism is enabled, the Stage-Level manager sends information to a single rank in the elastic instance, which then broadcasts the information to other ranks using NCCL.

For each elastic instance, ElasticMM manages the KV cache pool using PagedAttention at the granularity of a single token. Referencing previous work~\citep{zheng2024sglang}, we employ an LRU strategy for dynamic release in the cache pool. Specifically, each KV cache node in the prefix tree maintains a user count, and when this count drops to zero, it becomes eligible for eviction. As GPU memory quickly fills with KV cache, the eviction manager releases KV nodes based on least-recently-used order when the cache pool reaches its limit.

Communication between elastic instances is based on NCCL using dedicated CUDA streams, with hardware support from NVLINK offering 400GB/s bandwidth. To support multiple dynamic parallel groups at the iteration level, we use NCCL group functions to merge multiple point-to-point operations, forming collective operations in selected NCCL ranks.

\section*{Appendix B: Mathematical Proof of Inference Equivalence in Elastic Multimodal Parallelism}

In this appendix, we provide a formal proof that the Elastic Multimodal Parallelism (EMP) framework preserves inference equivalence with the standard inference process, ensuring that our system optimizations do not affect model accuracy.

\subsection*{B.1 Preliminaries and Notation}

We begin by formalizing the inference process in Multimodal Large Language Models (MLLMs). Let us define:

\begin{itemize}
    \item $\mathcal{M}$: A multimodal large language model
    \item $\mathcal{I}$: Set of image inputs
    \item $\mathcal{T}$: Set of text inputs
    \item $\mathcal{O}$: Set of text outputs
    \item $f_{\mathcal{M}}: \mathcal{I} \times \mathcal{T} \rightarrow \mathcal{O}$: The inference function of model $\mathcal{M}$
\end{itemize}

For a standard inference process, the computation can be decomposed into three sequential stages:

\begin{itemize}
    \item $g_E: \mathcal{I} \rightarrow \mathcal{V}$: Encoding function that maps images to visual tokens $\mathcal{V}$
    \item $g_P: \mathcal{V} \times \mathcal{T} \rightarrow \mathcal{H} \times \mathcal{K}$: Prefill function that produces hidden states $\mathcal{H}$ and KV cache $\mathcal{K}$
    \item $g_D: \mathcal{H} \times \mathcal{K} \rightarrow \mathcal{O}$: Decoding function that generates output tokens
\end{itemize}

\subsection*{B.2 Inference Equivalence Theorem}

\begin{theorem}[Inference Equivalence]
For any multimodal model $\mathcal{M}$ with inference function $f_{\mathcal{M}}$, the Elastic Multimodal Parallelism framework produces outputs identical to the standard sequential execution, i.e.,
\begin{equation}
f_{\mathcal{M}}(\mathcal{I}, \mathcal{T}) = g_D(g_P(g_E(\mathcal{I}), \mathcal{T})) = f_{\mathcal{M}}^{EMP}(\mathcal{I}, \mathcal{T})
\end{equation}
where $f_{\mathcal{M}}^{EMP}$ represents the inference function under the EMP framework.
\end{theorem}

\begin{proof}
We prove this by examining each component of our EMP framework and demonstrating that they maintain computational equivalence.
\end{proof}

\subsection*{B.3 Modality-Level Equivalence}

\begin{lemma}[Modality-Level Equivalence]
The separation of requests into modality groups preserves inference equivalence.
\end{lemma}

\begin{proof}
In EMP, we partition requests into text-only ($\mathcal{R}_T$) and multimodal ($\mathcal{R}_M$) groups. For any request $r \in \mathcal{R}_T$, the computation involves only $g_P$ and $g_D$ on text inputs, while for $r \in \mathcal{R}_M$, it involves the complete pipeline $g_E$, $g_P$, and $g_D$.

For any request $r = (i, t)$ where $i \in \mathcal{I}$ and $t \in \mathcal{T}$:
\begin{enumerate}
    \item If $i = \emptyset$ (text-only request), then $f_{\mathcal{M}}(i, t) = g_D(g_P(\emptyset, t))$.
    \item If $i \neq \emptyset$ (multimodal request), then $f_{\mathcal{M}}(i, t) = g_D(g_P(g_E(i), t))$.
\end{enumerate}

Our modality-aware load balancer ensures requests are routed to appropriate groups without altering their computation path. Therefore, for any request $r$, the output remains identical regardless of group assignment.
\end{proof}

\subsection*{B.4 Stage-Level Equivalence}

\begin{lemma}[Inference Stage Separation]
The decoupling of encoding, prefill, and decoding stages across separate computational resources preserves computational equivalence.
\end{lemma}

\begin{proof}
In the standard sequential execution, a multimodal request $r = (i, t)$ undergoes:
\begin{equation}
o = g_D(g_P(g_E(i), t))
\end{equation}

EMP disaggregates this pipeline into independently scheduled stages that may execute on different hardware instances:

\begin{align}
v &= g_E(i) & \text{(Encoding stage on instance $E$)} \\
(h, k) &= g_P(v, t) & \text{(Prefill stage on instance $P$)} \\
o &= g_D(h, k) & \text{(Decode stage on instance $D$)}
\end{align}

For this disaggregation to preserve equivalence, we must ensure:

\begin{enumerate}
    \item \textbf{Lossless intermediate representation}: The visual tokens $v$ generated by $g_E$ must be identical when transferred from instance $E$ to instance $P$. This is guaranteed by our use of deterministic serialization and deserialization protocols with checksum verification.
    
    \item \textbf{Computational state preservation}: The hidden states $h$ and KV cache $k$ generated by $g_P$ must be identical when transferred from instance $P$ to instance $D$. Our implementation uses exact memory copying with NCCL primitives that ensure bit-level accuracy during transfers.
    
    \item \textbf{Execution determinism}: Each function $g_E$, $g_P$, and $g_D$ must produce identical outputs for identical inputs regardless of hardware allocation. This is ensured by:
    \begin{itemize}
        \item Using deterministic CUDA operations
        \item Fixing random seeds across all instances
        \item Employing identical floating-point computation settings
    \end{itemize}
\end{enumerate}

Therefore, by induction on the stages:
\begin{align}
v_{EMP} &= v_{standard} \\
(h_{EMP}, k_{EMP}) &= (h_{standard}, k_{standard}) \\
o_{EMP} &= o_{standard}
\end{align}

Thus, the decoupled execution preserves computational equivalence with the standard sequential execution.
\end{proof}

\begin{lemma}[Dynamic Parallelism Invariance]
Changes in the degree of parallelism within each inference stage do not affect output correctness.
\end{lemma}

\begin{proof}
ElasticMM dynamically adjusts parallelism strategies for different inference stages. Let $\Pi_E$, $\Pi_P$, and $\Pi_D$ represent different parallelism configurations for the encoding, prefill, and decoding stages respectively.

For the encoding stage $g_E$ executed under parallelism strategy $\Pi_E$, we denote the execution as $g_E^{\Pi_E}$. We must prove that:

\begin{equation}
g_E^{\Pi_E^1}(i) = g_E^{\Pi_E^2}(i)
\end{equation}

for any image input $i$ and any two valid parallelism strategies $\Pi_E^1$ and $\Pi_E^2$.

Our implementation primarily uses data parallelism, where computation is partitioned across multiple devices and results are gathered using deterministic reduction operations. For data parallelism with $n$ devices:

\begin{equation}
g_E^{DP_n}(i) = \text{Gather}(\{g_E^1(i_1), g_E^2(i_2), \ldots, g_E^n(i_n)\})
\end{equation}

where $i_1, i_2, \ldots, i_n$ represent partitions of input $i$, and $g_E^j$ represents the computation on device $j$.

Since the Gather operation performs deterministic aggregation (using synchronous all-reduce operations with fixed-order reduction), the parallelism strategy does not affect the mathematical result. The same property holds for $g_P$ and $g_D$.

For encoder-decoder architectures with cross-attention mechanisms, we ensure that the cross-attention patterns remain identical regardless of parallelism strategy by maintaining consistent attention mask distributions across devices.

Therefore, for any valid parallelism strategies:
\begin{align}
g_E^{\Pi_E^1}(i) &= g_E^{\Pi_E^2}(i) \\
g_P^{\Pi_P^1}(v, t) &= g_P^{\Pi_P^2}(v, t) \\
g_D^{\Pi_D^1}(h, k) &= g_D^{\Pi_D^2}(h, k)
\end{align}

Thus, the output of the entire inference process remains invariant to changes in parallelism strategy.
\end{proof}

\subsection*{B.5 Data Integrity During Migration}

\begin{lemma}[KV Cache Migration Fidelity]
KV cache migration during elastic scaling preserves computational equivalence.
\end{lemma}

\begin{proof}
When ElasticMM performs instance scaling, it migrates KV cache entries $k \in \mathcal{K}$ between GPUs. Let $\mathcal{K}_s$ represent the KV cache on source instance $s$ and $\mathcal{K}_d$ represent the same cache after migration to destination instance $d$.

We must show that $\mathcal{K}_s = \mathcal{K}_d$ after migration. Our system implements exact copying of memory blocks using NCCL communications with error checking. For each tensor $T \in \mathcal{K}_s$, the migration process performs:

\begin{equation}
T_d = \text{Copy}(T_s)
\end{equation}

Since modern GPU interconnects (NVLink) support lossless data transfer and our implementation verifies integrity through checksums, we ensure $T_s = T_d$ for all tensors, thus $\mathcal{K}_s = \mathcal{K}_d$.

To formalize this further, let $\mathcal{K} = \{T_1, T_2, \ldots, T_m\}$ be the set of tensors in the KV cache. We define a migration function $\mu: \mathcal{K}_s \rightarrow \mathcal{K}_d$. For each tensor $T_i \in \mathcal{K}_s$:

\begin{equation}
\mu(T_i) = T_i + \epsilon_i
\end{equation}

where $\epsilon_i$ represents any potential error introduced during migration.

Our implementation guarantees that:
\begin{equation}
\|\epsilon_i\|_{\infty} = 0 \quad \forall i \in \{1, 2, \ldots, m\}
\end{equation}

Therefore, $\mu(T_i) = T_i$ for all $i$, ensuring $\mathcal{K}_s = \mathcal{K}_d$.

After migration, the decoding process continues with the exact same KV cache state, thus producing identical outputs to the non-migrated scenario.
\end{proof}

\subsection*{B.6 Non-blocking Encoding Equivalence}

\begin{lemma}[Non-blocking Encoding Correctness]
The non-blocking encoding optimization preserves inference output equivalence.
\end{lemma}

\begin{proof}
In standard sequential execution, the encoding process $g_E$ blocks further computation until visual tokens are generated. In ElasticMM's non-blocking encoding implementation, the encoding process executes asynchronously in parallel with other computations.

Let $r = (i, t)$ be a multimodal request. In standard execution:
\begin{equation}
o = g_D(g_P(g_E(i), t))
\end{equation}

With non-blocking encoding, we have:
\begin{align}
v &= g_E(i) \quad \text{(executes asynchronously)} \\
(h, k) &= g_P(v, t) \quad \text{(waits for $v$ to be available)} \\
o &= g_D(h, k)
\end{align}

Because our implementation ensures proper synchronization before the prefill stage accesses the encoded visual tokens, the data dependencies are preserved. Specifically, the prefill stage $g_P$ will not begin execution until the encoding result $v = g_E(i)$ is complete and available.

The non-blocking optimization affects only the scheduling of operations across compute resources, not the mathematical computations themselves. All data dependencies in the computational graph are preserved through synchronization barriers that ensure the prefill stage has access to the complete and correctly encoded visual tokens.

Therefore, non-blocking encoding preserves inference equivalence while improving computational efficiency.
\end{proof}

\subsection*{B.7 Unified Multimodal Prefix Caching Correctness}

\begin{lemma}[Prefix Cache Correctness]
The unified multimodal prefix caching mechanism preserves inference equivalence.
\end{lemma}

\begin{proof}
Our unified multimodal prefix caching mechanism stores and reuses computed results for both visual encodings and KV cache prefixes. For any request $r = (i, t)$, we maintain a cache that maps inputs to their computed representations:

\begin{align}
C_V &: \mathcal{I} \rightarrow \mathcal{V} \quad \text{(Visual token cache)} \\
C_K &: \mathcal{V} \times \mathcal{T} \rightarrow \mathcal{K} \quad \text{(KV cache prefix)}
\end{align}

For cached visual tokens, we must show that:
\begin{equation}
\forall i \in \mathcal{I}: C_V(i) = g_E(i)
\end{equation}

This is guaranteed by our deterministic image preprocessing and encoding pipeline, which ensures that identical inputs produce identical encoded outputs. We use cryptographic hashing to verify input identity.

For cached KV prefixes, we must show that:
\begin{equation}
\forall (v, t_{\text{prefix}}) \in \mathcal{V} \times \mathcal{T}: C_K(v, t_{\text{prefix}}) = g_P^{\text{partial}}(v, t_{\text{prefix}})
\end{equation}

where $g_P^{\text{partial}}$ computes KV cache entries for the prefix portion of the text.

When a cache hit occurs, ElasticMM reuses the cached computations as follows:
\begin{align}
v &= \begin{cases}
C_V(i) & \text{if $i$ is in cache} \\
g_E(i) & \text{otherwise}
\end{cases} \\
k_{\text{prefix}} &= \begin{cases}
C_K(v, t_{\text{prefix}}) & \text{if $(v, t_{\text{prefix}})$ is in cache} \\
g_P^{\text{partial}}(v, t_{\text{prefix}}) & \text{otherwise}
\end{cases} \\
(h, k_{\text{full}}) &= g_P^{\text{remaining}}(v, t, k_{\text{prefix}}) \\
o &= g_D(h, k_{\text{full}})
\end{align}

Since cached values are exact duplicates of what would be computed from scratch, and our cache invalidation logic ensures stale entries are never used, the output $o$ remains identical to non-cached execution for any given input.
\end{proof}

\subsection*{B.8 Analysis of Numerical Stability}

While the mathematical equivalence is guaranteed in theory, practical implementations may introduce minor numerical differences due to floating-point operations. We now analyze these potential sources of error.

\begin{proposition}[Numerical Stability]
Any numerical differences introduced by EMP are bounded and negligible.
\end{proposition}

\begin{proof}
The primary sources of potential numerical differences in EMP are:

\begin{enumerate}
    \item \textbf{Parallel computation order}: When using data parallelism, the order of reduction operations could theoretically affect floating-point summation due to non-associativity. However, modern frameworks use deterministic reduction algorithms that ensure consistent results regardless of partition count.
    
    \item \textbf{Tensor partitioning boundaries}: In some parallel strategies, tensor partitioning might introduce different computation patterns. ElasticMM prioritizes data parallelism which preserves tensor integrity.
    
    \item \textbf{Mixed precision operations}: When using mixed precision, the accumulation of partial results might vary slightly. However, these differences are typically on the order of $\epsilon \approx 10^{-7}$ for fp16 operations, which is far below the threshold of affecting logical model outputs.
\end{enumerate}

For token generation specifically, let $p(w_t | w_{<t})$ be the probability of generating token $w_t$ given previous tokens. The maximum variation in these probabilities due to numerical differences is bounded by:

\begin{equation}
|\Delta p(w_t | w_{<t})| < \epsilon_{\max}
\end{equation}

where $\epsilon_{\max} \approx 10^{-7}$ for fp16 operations.

Since token selection uses argmax operations, these minute differences do not affect the final output unless two candidate tokens have probability differences smaller than $\epsilon_{\max}$, which is statistically negligible.
\end{proof}

\subsection*{B.9 Empirical Validation}

To empirically verify our theoretical guarantees, we conducted an experiment comparing outputs from standard sequential inference and EMP-based inference across 1,000 diverse prompts from our evaluation datasets.

\begin{table}[h]
\centering
\caption{Output Consistency between Standard and EMP Inference}
\begin{tabular}{lcc}
\toprule
\textbf{Model} & \textbf{Identical Outputs (\%)} & \textbf{Avg. Token Probability Diff.} \\
\midrule
Qwen2.5-VL 7B & 100\% & $<10^{-8}$ \\
Llama3.2-Vision 11B & 100\% & $<10^{-8}$ \\
\bottomrule
\end{tabular}
\end{table}

The outputs were bit-identical in 100\% of cases, confirming that our EMP framework preserves inference equivalence in practice, including when stages are separated across different computational resources.

\subsection*{B.10 Conclusion}

We have formally proven that Elastic Multimodal Parallelism maintains exact inference equivalence with standard sequential execution. Our proof demonstrates that:

\begin{enumerate}
    \item Separation into modality groups preserves computational paths
    \item Decoupling of inference stages across different resources maintains output equivalence
    \item Dynamic adjustment of parallelism strategies does not affect results
    \item KV cache migration during elastic scaling preserves state fidelity
    \item Non-blocking encoding and unified prefix caching optimizations maintain correctness
\end{enumerate}

This mathematical guarantee ensures that all performance improvements reported in our experimental evaluation come without any sacrifice in model accuracy or output quality. ElasticMM therefore achieves superior efficiency while maintaining the exact same inference results as traditional sequential execution.

\end{document}